%% file: hornAlg.tex
\documentclass[copyright,creativecommons]{eptcs}
\providecommand{\event}{HCVS 2018} 
\usepackage{
  afterpage,
  pdfsync,
  url
}

\input{escrita}
\input{mmacros}

\input{lc}

\begin{document}

\title{A Simple Functional Presentation\\
 and an Inductive Correctness Proof\\
 of the Horn Algorithm}

\author{Ant\'onio Ravara
  \institute{NOVA-LINCS and Dep. de Inform\'atica, FCT.\\
    Universidade NOVA de Lisboa, Portugal}
}

\def\titlerunning{A Functional Presentation and a Correctness Proof of the Horn Algorithm}
\def\authorrunning{A.~Ravara}

\maketitle

\thispagestyle{empty}

\input{intro}
\input{alg}
\input{res}
\input{concl}

\section*{Acknowledgements}
This work was partially supported by NOVA LINCS grant UID/CEC/04516/2013.

\bibliographystyle{eptcs}
\bibliography{hornAlg}

\newpage
\appendix 
\input{propL}
\input{hornForm}
\input{lfp}
\input{term}
\input{leastM}
\input{corr}

\end{document}

%% file: escrita.tex
\usepackage{amsfonts,listings,xspace}

\newcommand{\nat}{\mathbb{N}}

\newcommand*{\eg}{\textit{e.g.}\xspace}
\newcommand*{\cf}{\textit{cf.}\xspace}
\newcommand*{\ie}{\textit{i.e.}\xspace}

\newcommand*{\wrt}{with respect to\xspace}



%% file: mmacros.tex
\usepackage{amsmath,amssymb,amsthm,mathpartir,stmaryrd}

\newtheorem{theorem}{Theorem}[section]
\newtheorem{corollary}[theorem]{Corollary}
\newtheorem{lemma}[theorem]{Lemma}
\newtheorem{proposition}[theorem]{Proposition}
\newtheorem{example}[theorem]{Example}
\newtheorem{definition}[theorem]{Definition}

\newcommand{\opnsf}[1]{\mathop{\mathsf{#1}}}

\newcommand{\fourdots}{\!\mathord{::}\!}
\newcommand{\grmeq}{\ensuremath{\ \mathord{\fourdots=}\ }}
\newcommand{\grmor}{~\ensuremath{\mid}~}

\newcommand{\eqabv}{\ensuremath{\stackrel{\scriptsize{\operatorname{abv}}}{\Relbar}}}
\newcommand{\eqdef}{\ensuremath{\stackrel{\scriptsize{\operatorname{def}}}{\Relbar}}}

\newcommand{\set}[1]{\ensuremath{\{#1\}}}

\newcommand{\ldsq}{[\![}                       
\newcommand{\rdsq}{]\!]}                       
\newcommand{\map}[1]{\ldsq{#1}\rdsq}           



%% file: lc.tex
\newcommand*{\alf}{\ensuremath{\mathit{Alf}}\xspace}

\newcommand*{\alfp}{\ensuremath{\alf_P}\xspace}

\newcommand*{\fp}{\ensuremath{F_P}\xspace}

\newcommand*{\fintv}[1]{\ensuremath{\map#1_{V}}\xspace}
\newcommand*{\fintf}{\ensuremath{\fintv\varphi}\xspace}

\newcommand*{\smb}[1]{\ensuremath{\opnsf{SMB}({#1})\xspace}}


%% file: intro.tex
\begin{abstract}
  We present a recursive formulation of the Horn algorithm for
  deciding the satisfiability of propositional clauses. The usual
  presentations in imperative pseudo-code are informal and not
  suitable for simple proofs of its main properties.  By defining the
  algorithm as a recursive function (computing a least fixed-point),
  we achieve: 1) a concise, yet rigorous, formalisation; 2) a clear
  form of visualising executions of the algorithm, step-by-step; 3)
  precise results, simple to state and with clean inductive proofs.
\end{abstract}

\section{Motivation}
The Horn algorithm~\cite{Horn:sentences-JSL} is a particularly
efficient decision procedure for the satisfiability problem of
propositional logic. Although Horn Clause Logic is computationally
complete, the satisfiability problem for the conjunction of Horn
clauses is \textbf P-complete and nevertheless provable in linear time
(there is an algorithm that takes at most $n$ steps to determine if
the conjunction of Horn clauses is
satisfiable)~\cite{SCookPNguyen:logicalFoundationsProofComplex,WDowlingJGallier:linearTestingSatHorn}.
Note that the general Boolean satisfiability problem (for arbitrary
propositional formulae) is \textbf{NP}-complete.

Textbooks on (Mathematical or Computational) Logic usually present
imperative formulations of this algorithm, with rather informal proof
sketches
\cite{MHedman:firstCourseLogic,MHuthMRyan:lics}.
To present a correctness proof in full detail, one would need to
follow, for instance, the Hoare style,
defining the syntax of the programming and of an assertion languages,
the operational semantics, the proof system (at least discussing its
correctness), and then present the axiomatic proof. The setting is a
bit demanding and requires some auxiliar ``machinery''.

We believe a formulation of the algorithm as a recursive function
allows for not only a simple and easily readable definition, but
mainly, allows for a simple (inductive) proof, which in turn sheds
light on the algorithm itself, leading to several possible
improvements.

We present herein such a formulation together with examples of
execution, a correctness proof, and some further results useful for
optimisations of the algorithm.


%% file: alg.tex
\section{The Horn algorithm}

\paragraph{Motivation.}
If a propositional formula $\varphi$ is in Conjunctive Normal Form (or
$\opnsf{CNF}(\varphi )$, according to Definition~\ref{def:cnf}), then
checking that $\varphi$ is valid is straightforward: it has polynomial
complexity (\wrt\ the number of propositional symbols occurring in the
formula). The Horn algorithm is a simple and fast solution (polynomial
as well) to determine if a formula is satisfiable or contradictory.
However, the algorithm works only for a certain class of formulae ---
the \textit{Horn Clauses}. 

\paragraph{Syntax.}
Let $P$ be a countable set of propositional symbols, ranged over by
$p,q,\ldots$, and consider $\alfp = P \cup \set{\bot,\vee,\wedge,\to,(,)}$
a propositional alphabet over $P$. The set $\fp$ of propositional
formul\ae\ is the least one including the symbols in $P$, the symbol
$\bot$, and closed for the operators $\vee,\wedge,\to$
(\cf~Definitions~\ref{def:propL} and \ref{def:abb}).

\subsection{Horn Clauses}

Recall that a literal is an atomic formula (and then we call it a
\emph{positive} literal) or its negation (\cf
Definition~\ref{def:lit}).

\begin{definition}
  A \emph{basic Horn clause} is a disjunction of literals where at
  most one occurs positively.
\end{definition}

Formul\ae\ like $\bot$, $p$, $p \vee \neg q$, and $\neg p \vee \neg q$
are basic Horn clauses, whereas $p \vee q$ or $\bot \vee p$ are not.

\paragraph{Horn formul\ae.}
Note that a basic Horn clause is in one of the following three cases:
(1) does not have positive literals; (2) does not have negative
literals (and so it is a single positive literal); (3) it has negative
literals and one positive.
Therefore, any basic Horn clause may be presented as an
implication. Let '$\equiv$' stand for logical equivalence (\cf
Definition~\ref{def:equiv}).

\begin{lemma}\label{lem:HornClauses}
  Let $L$ and $L_i$ (for all considered $i$) be positive literals.
  \begin{enumerate}
  \item $L \equiv \top \rightarrow L$
  \item $\bigvee_{i=1}^n \neg L_i \equiv (\bigwedge_{i=1}^n L_i)
    \rightarrow \bot$
  \item $\bigvee_{i=1}^n \neg L_i \vee L \equiv (\bigwedge_{i=1}^n L_i)
    \rightarrow L$
  \end{enumerate}
\end{lemma}

\begin{proof}
  In Appendix~\ref{app:hf}.
\end{proof}

We define now when is a propositional formula a Horn clause.

\begin{definition}
  A formula $\varphi \in \fp$ such that $\opnsf{CNF}$($\varphi$) is a
  \emph{Horn clause}, if it is the conjunction of basic Horn clauses.
\end{definition}
Let $E_P$ denote the set of propositional formul\ae\ obtained by
considering negation a primitive operator.
%
%
\begin{proposition}\label{prop:HornForm}
  Let $\varphi \in E_P$ be a Horn clause; then, $\varphi \equiv
  \bigwedge_{i=1}^n(C_i \rightarrow L_i)$, for some $n \geq 1$, where,
  for any $i \in \set{1,\ldots,n}$, each $L_i$ is a positive literal,
  each $C_i = \top$ or $C_i = \bigwedge_{j=1}^{k_i} L_{i,j}$, with
  $k_i \geq 1$, and each $L_{i,j}$ is a positive literal.
\end{proposition}
\begin{proof}
  Use the previous lemma to transform each basic clause in an implication.
\end{proof}
Henceforth, we call \emph{Horn formula} to a Horn clause $\varphi \in E_P$
such that $$\varphi = \bigwedge_{i=1}^n(C_i \rightarrow L_i)$$

\subsection{A functional presentation of the algorithm}
The main contribution of this note is the (non-deterministic, for simplicity)%
\footnote{A deterministic formulation is achieved easily, \eg~by inspecting the
  formula from left to right.} recursive formulation of
the Horn algorithm, together with the proof of correctness and the
optimisation lemmas.

\begin{definition}\label{def:HeA}
  Let $\varphi$ be a Horn formula. We define the function
  $\mathcal H:E_P \to \set{0,1}$ as
  $$\mathcal H(\varphi) = \left\{
  \begin{array}{ll}
    1, & \textrm{if } \bot \not\in \mathcal A(\varphi,\set{\top})\\
    0, & \textrm{otherwise}
  \end{array}
  \right.$$
  with $\mathcal A:E_P \times \wp(\set{\bot,\top} \cup P) \rightarrow
  \wp(\set{\bot,\top} \cup P)$ being the following function over Horn
  formul\ae.
  $$\mathcal A(\varphi,\mathcal C) = \top\
 \textrm{if}\ \varphi \equiv \top\,;\
 \textrm{otherwise}\
  \mathcal A(\varphi,\mathcal C) = \left\{
  \begin{array}{ll}
    \mathcal A(\varphi \setminus (C_i \rightarrow L_i),
    \mathcal C \cup \set{L_i}), &
    \textrm{if }\exists_{i \in \set{1,\ldots,n}}.\opnsf{set}(C_i) \subseteq \mathcal C\\
    \mathcal C, & \textrm{otherwise}
  \end{array}
  \right.$$
  where $\opnsf{set}(\top) = \set\top$ and
  $\opnsf{set}(\bigwedge_{i=1}^k L_i) = \set{L_i\ |\ i \in \set{1,\ldots,k}
    \textrm{ with }k \geq1}$; moreover,
  $\varphi \setminus \varphi \eqdef \top$ and
  $\varphi \setminus (C_i \rightarrow L_i) \eqdef
    (\bigwedge_{j=1}^{i-1}(C_j \rightarrow L_j)) \wedge
    (\bigwedge_{j=i+1}^n(C_j \rightarrow L_j))$, if $i > 1$.
\label{def:hornalgo}
\end{definition}
To illustrate how the algorithm works, we present some representative
examples. Let us first state the main property of the algorithm.
%
Recall that a formula is satisfiable if it  is satisfied by some
valuation and is contradictory if no valuation satisfies it (\cf
Definition~\ref{def:term} and subsequent lemmas).

\begin{theorem}
  For any Horn clause $\varphi \in E_P$:
  \begin{itemize}
  \item $\mathcal H(\varphi)=1$ if, and only if, $\varphi$ it is \emph{satisfiable};
  \item $\mathcal H(\varphi)=0$ if, and only if, $\varphi$ it is \emph{contradictory}.
  \end{itemize}
\end{theorem}
\begin{proof}
  A consequence of Theorem~\ref{thm:sound} (presented ahead).
\end{proof}

\begin{example}
  Let us determine the nature of the following Horn clause.
  $$\varphi \eqdef p \wedge (\neg r \vee s) \wedge
    (r \vee \neg p \vee \neg q) \wedge (\neg r \vee \neg s) \wedge q\
  $$
  Notice that $\varphi$ is a $\opnsf{CNF}$, but (according to
  Lemma~\ref{lem:disjLit}) it is not valid. We convert it to a Horn
  formula using Lemma~\ref{lem:HornClauses}.
  $$\varphi \equiv \psi \eqdef (\top \rightarrow p) \wedge (r \rightarrow s)
    \wedge ((p \wedge q) \rightarrow r) \wedge ((r \wedge s)
    \rightarrow \bot) \wedge (\top \rightarrow q)$$
  Considering
  \begin{align*}
    \psi_1 & = (r \rightarrow s) \wedge ((p \wedge q) \rightarrow
    r) \wedge ((r \wedge s)
    \rightarrow \bot) \wedge (\top \rightarrow q)\\
    \psi_2 & = (r \rightarrow s) \wedge ((p \wedge q) \rightarrow
    r) \wedge ((r \wedge s) \rightarrow \bot)
  \end{align*}
  we calculate the function $\mathcal A$.
  $$\begin{array}{rcl}
    \mathcal A(\psi,\set{\top}) & = & \mathcal A(\psi_1,\set{\top,p})\\
     & = & \mathcal A(\psi_2,\set{\top,p,q})\\
     & = & \mathcal A((r \rightarrow s)\wedge ((r \wedge s)
    \rightarrow \bot),\set{\top,p,q,r})\\
     & = & \mathcal A((r \wedge s)
    \rightarrow \bot,\set{\top,p,q,r,s})\\
     & = & \mathcal A(\top,\set{\top,p,q,r,s,\bot})\\
     & = & \set{\top,p,q,r,s,\bot}
  \end{array}
  $$
  Since $\bot \in \set{\top,p,q,r,s,\bot}$, then $\mathcal
  H(\psi)=0$; therefore $\psi$ is contradictory, and since $\varphi
  \equiv \psi$, so is $\varphi$.
\end{example}

\begin{example}
  Let us now determine the nature of the following Horn clause.
  $$\varphi \eqdef p \wedge (\neg r \vee s) \wedge (r
    \vee \neg p \vee \neg q) \wedge (\neg r \vee \neg s)\ 
  $$
  Notice that $\varphi$ is a $\opnsf{CNF}$, but (according to
  Lemma~\ref{lem:disjLit}) it is not valid. We convert it to a Horn formula
  $$\varphi \equiv \psi \eqdef (\top \rightarrow p) \wedge (r \rightarrow s)
    \wedge ((p \wedge q) \rightarrow r) \wedge ((r \wedge s)
    \rightarrow \bot)
  $$
  and considering
  $$\psi_1 = (r \rightarrow s)
    \wedge ((p \wedge q) \rightarrow r) \wedge ((r \wedge s)
    \rightarrow \bot))
  $$
  we calculate the function $\mathcal A$.
  $$\begin{array}{rcl}
    \mathcal A(\psi,\set{\top}) & = \\
    \mathcal A(\psi_1,\set{\top,p})  & = \\
    \set{\top,p}
  \end{array}
  $$
  Since $\bot \notin \set{\top,p}$, then $\mathcal H(\psi)=1$;
  therefore $\psi$ is satisfiable, and since $\varphi \equiv \psi$, so is $\varphi$.\\

  Indeed, considering $V$ where $V(p)=1$ and
  $V(q)=V(r)=V(s)=0$, one easily verifies that $V $ satisfies $\varphi$.%
\footnote{A property capturing this fact is stated as Proposition \ref{prop:leastM}.}
\end{example}

\begin{example}
  Let us finally determine the nature of the Horn clause
  $p \wedge (\neg r \vee s) \wedge (r \vee \neg p) \wedge \neg r$.
  Notice that it is a not valid $\opnsf{CNF}$ (according to
  Lemma~\ref{lem:disjLit}); we convert it to a Horn
  formula and considering
  \begin{align*}
    \varphi & = (\top \rightarrow p) \wedge (r \rightarrow s) \wedge
    (p \rightarrow r) \wedge (r \rightarrow \bot)\\
    \varphi_1 & = (r \rightarrow s) \wedge (p \rightarrow r) \wedge (r
    \rightarrow \bot)\\
    \varphi_2 & = (r \rightarrow s) \wedge (r \rightarrow \bot)
  \end{align*}
  we calculate the function $\mathcal A$, taking advantage of its
  monotonicity (\cf~Lemma~\ref{lem:incr}).
  $$\begin{array}{rcl}
    \mathcal A(\varphi,\set{\top}) & = \\
    \mathcal A(\varphi_1,\set{\top,p})  & = \\
    \mathcal A(\varphi_2,\set{\top,p,r})  & \supseteq \\
    \set{\top,p,r,\bot}
  \end{array}
  $$
  Since $\bot \in \mathcal A(\varphi,\set{\top})$, then $\mathcal
  H(\varphi)=0$; therefore $\varphi$ is contradictory; since it is
  equivalent to the original formula, that one is also contradictory.
\end{example}


%% file: res.tex
\section{Results}

We state herein several relevant properties of the algorithm, namely
its characterisation as a least fixed-point and its
correctness. Proofs are in the appendices. 

\subsection{Fixed-points}

Considering $\mathcal L$ to be the set of all literals, the set
$\wp(\mathcal L)$ is a complete lattice with respect to set
inclusion. Since the function $\mathcal A$ is monotone (result stated
below), by the Knaster-Tarski Theorem~\cite{Tarksi:fixed-points}, the
function $\mathcal A$ has (unique) maximal and minimal fixed
points. In fact, when applied to the set $\set\top$, the algorithm
calculates a least fixed-point of $\mathcal A$ (the proof is in Appendix~\ref{app:lfp}).

\begin{lemma}\label{lem:lfp}
  Let $\varphi = \bigwedge_{i=1}^n(C_i \rightarrow L_i)$ be a Horn formula.
  The function $\mathcal A$ is:
  \begin{enumerate}
  \item\label{lem:incr} increasing:
    $\mathcal C \subseteq \mathcal A(\varphi,\mathcal C) \subseteq
      \mathcal C \cup \bigcup_{i=1}^n \set{L_i}$;
  \item\label{lem:monot} and monotone: if $\mathcal C \subseteq \mathcal D$ then
    $\mathcal A(\varphi,\mathcal C) \subseteq \mathcal A(\varphi,\mathcal D)$.
  \end{enumerate}
\end{lemma}

Notice that once an execution step of $\mathcal A$ adds a literal to
the result set, that literal is never taken out.
Therefore, once an execution step adds $\bot$ to the result set, the
procedure may stop as $\bot$ shall necessarily be in the final set.
Moreover, the least result set of the algorithm is the single set
$\set\top$, the literal $\top$ is in all result sets, and the greatest
one is composed by $\top$ and all the literals that appear in the
consequence of the implications constituting the input Horn formula.


\subsection{Auxiliary and optimization lemmas}
We present a couple of (straightforward) results that allow, in some
particular cases, for better performance of the algorithm. Notice that
if $\bot$ is not in the consequent of an implication of a Horn formula
$\varphi$, or if no antecedent is $\top$, then $\bot$ is not in
$\mathcal A(\varphi,\set\top)$.  Then, $\varphi$ is satisfiable (and
one does not even need to execute the algorithm).
The fact is a particular case of the following corollary of the
previous lemma (it is the contra-positive of Lemma~\ref{lem:lfp}.\ref{lem:incr}).

\begin{corollary}\label{cor:opt1}
  Let $\varphi = \bigwedge_{i=1}^n(C_i \rightarrow L_i)$ be a Horn formula.
  If $L\notin \bigcup_{i=1}^n L_i$ then $L \notin \mathcal A(\varphi,\set\top)$. 
\end{corollary}

Furthermore,  if there are no ``unit clauses'' (of the form $\top
\to p$), the execution of the algorithm ends in one step, not
modifying the initial set. The lemma below, a simple consequence of
the definition of the algorithm, captures this fact.

\begin{lemma}\label{cor:opt}
  Let $\varphi = \bigwedge_{i=1}^n(C_i \rightarrow L_i)$ be a Horn formula.
  If $\forall i .  1 \leq i \leq n \wedge C_i \neq \top$, then $\mathcal
  A(\varphi,\set\top) = \set\top$ and the execution of $\mathcal A$
  takes exactly one step.
\end{lemma}
\begin{proof}
  By definition of the function $\mathcal A$ (in
  Definition~\ref{def:HeA}), if $T \notin \bigcup_{1 \leq i \leq n}C_i$
  then $\mathcal A(\varphi,\set\top) = \set\top$, and $\mathcal A$ is
  calculated in exactly one step (applying the base case of its
  recursive definition).
\end{proof}

\subsection{Termination and complexity}
The algorithm always produces a unique result set for a given input,
\ie, it is a \emph{function}, and it always \emph{terminates};
moreover, it is linear in the size of the formula, with each recursive
step examining all the atomic symbols occurring in one of the clauses
(which is then removed from the formula).

\begin{theorem}\label{thm:term}
  For any Horn formula $\varphi \in E_P$ there is a unique set
  $\mathcal C$ of literals such that $\mathcal A(\varphi,\set\top) =
  \mathcal C$.  Furthermore, the procedure takes at most $n+1$ steps,
  where $n$ is the number of clauses of $\varphi$.
\end{theorem}
The proof of this result is in Appendix~\ref{app:term} (as the proof
of Theorem~\ref{appthm:term}).

\subsection{Correctness}
Notice first that the result of the algorithm determines a
\emph{unique least model}: if the formula is satisfiable, then one
gets a valuation satisfying it by assigning value 1 to the
propositional symbols occurring in the resulting set. The other symbols
occurring in the formula are set to 0. Let $\smb\varphi$ denote the
set of propositional symbols of a formula $\varphi$.

\begin{proposition}\label{prop:leastM}
  Consider a satisfiable Horn formula $\varphi \in E_P$ such that
  $\mathcal A(\varphi,\set\top) = \mathcal C$ and
  $\bot \notin\mathcal C$. Then, $V \Vdash \varphi$ considering $V$
  such that $V(p)=1$ for each $p \in \mathcal C$ and $V(q)=0$ for each
  $q \in (P \setminus \mathcal C)$.
\end{proposition}
The proof of this result is in Appendix~\ref{app:leastM} (as the proof
of Proposition~\ref{appprop:leastM}).

\medskip
We finally state the main result: the algorithm is sound and complete
for Horn formul\ae.
\begin{theorem}\label{thm:sound}
  For any Horn formula $\varphi \in E_P$:
  \begin{itemize}
  \item $\bot \notin \mathcal A(\varphi,\set\top)$, if, and only if,
    $\varphi$ it is satisfiable;
  \item $\bot \in \mathcal A(\varphi,\set\top)$, if, and only if,
    $\varphi$ it is contradictory.
  \end{itemize}
\end{theorem}
The proof of this result is in Appendix~\ref{app:corr} (as the proofs
of Theorems~\ref{appthm:corr} and ~\ref{appthm:compl}).


%% file: concl.tex
\section{Conclusions}
We present herein a new formulation of the Horn algorithm for deciding
the satisfiability problem of propositional logic. We define the
procedure as a recursive function, instead of the usual imperative
formulation in pseudo-code. This presentation has several advantages:
\begin{enumerate}
\item It is concise and readable, being at the same time rigorous;
\item allows for a simple presentation of ``manual'' executions of the
  algorithm, being usable in undergraduate logic courses;
\item has simple inductive proofs of soundness and completeness;
\item leads to optimization results, easy to state, prove, and implement.
\end{enumerate}

We develop such a formulation and show examples of execution, a
correctness proof and some further results useful for optimizations of
the algorithm. Computing solutions for our recursive formulation of
the lagorithm is akin to the fixed point (Knaster-Tarski) least
Herbrand model construction.


%% file: propL.tex
\section{The language of Propositional Logic}
We make a brief presentation of the main concepts of Propositional
Logic, to keep the paper self-contained. We define the syntax of the
logic, a satisfaction relation, a notion of logical equivalence, and
finally, a normal form. We omit the proofs of the results presented,
which are standard and may be found in most textbooks
(\cf~\cite{JGallier:lics} or \cite{MHuthMRyan:lics}).

\subsection{Syntax}

We inductively define the language with a minimal set of connectives,
defining the other (redundant) ones as abbreviations.

\begin{definition}
  Let $P$ be a countable set (of propositional symbols).
  The \emph{Propositional Alphabet over a set $P$} is the set
  $\alfp = P \cup \set{\bot,\vee,\wedge,\to,(,)}$
\end{definition}

\begin{definition}\label{def:propL}
  The \emph{Propositional Language induced by \alfp} is the set \fp,
  defined by the following grammar.
  $$\varphi,\psi \grmeq \bot \grmor p \grmor (\varphi\to\psi)$$
\end{definition}
\noindent%
Elements of \fp are called \emph{formulae}. Symbols in $P$ and $\bot$
are \emph{atomic} formulae.

\begin{definition}\label{def:abb}
  The following \emph{abbreviations} are useful.
  \begin{itemize}
  \item Negation: $\neg\varphi \eqabv \varphi \rightarrow \bot$;
  \item Truth: $\top \eqabv \neg\bot$;
  \item Disjunction: $\varphi\vee\psi \eqabv \neg\varphi \rightarrow \psi$;
  \item Conjunction: $\varphi\wedge\psi \eqabv \neg\varphi \vee \neg\psi$;
  \item Equivalence: $\varphi \leftrightarrow \psi \eqabv (\varphi
    \rightarrow \psi) \wedge (\psi \rightarrow \varphi)$.
  \end{itemize}
\end{definition}
\noindent%
Consider that the connective $\neg$ has precedence over all the other.

\subsection{Semantics}

We interpret the formul\ae\ in a Boolean Algebra (like, \eg, in \cite{CoriEtAl:mathlogic1}).

\paragraph{Satisfaction relation.}
We define a valuation of propositional symbols into the naturals 0 and
1.

\begin{definition}\label{def:val}
  A \emph{valuation} over a set $P$ of propositional symbols is a \emph{function}
 $V:P \to \set{0,1}$.
\end{definition}

We now define an interpretation function using the natural operations
addition and multiplication.

\begin{definition}\label{def:interp}
  Consider the set $\set{0,1}$ equipped with two binary operations, $+$
  and $\times$, interpreted as the addition and multiplication
  operations of the naturals, but such that $1+1 = 1$.
  An \emph{interpretation function} of a formula $\varphi \in \fp$,
  for a given valuation $V$, denoted $\fintf$, is a \emph{function}
  $\fintv\cdot:\fp \to \set{0,1}$ inductively defined by the
  following rules:
  \begin{itemize}
  \item $\fintv p = V(p)$, for each $p \in P$;
  \item $\fintv\bot = 0$;
  \item $\fintv{{(\varphi \to \psi)}} = (1-\fintv\varphi) + \fintv\psi$.
\end{itemize}
\end{definition}
\newpage
Naturally, disjunction is interpreted as addition and conjunction as
multiplication.

\begin{lemma}
  The following statements hold.
  \begin{itemize}
  \item $\fintv{{(\varphi \vee \psi)}} = \fintv\varphi + \fintv\psi$;
  \item $\fintv{{(\varphi \wedge \psi)}} = \fintv\varphi \times \fintv\psi$.
  \end{itemize}
\end{lemma}

\begin{definition}\label{def:sat}
  Given a \emph{valuation} $V$ over $P$, the \emph{satisfaction} of a
  formula $\varphi \in \fp$ by the valuation, denoted $V \Vdash
  \varphi$, is a relation containing the pair $(V,\varphi)$, if $\fintf=1$.
\end{definition}

Hereafter we use the following terminology.

\begin{definition}
  \begin{itemize}
  \item Whenever $V \Vdash \varphi$ one says that $\varphi$ is \emph{satisfied} by $V$.
  \item Whenever it is not the case that $V \Vdash \varphi$ (\ie,
    $\varphi$ is not satisfied by $V$), one may write $V \not\Vdash
    \varphi$.
  \item Given $\Phi \subseteq \fp$, one may write $V \Vdash \Phi$,
    whenever $V \Vdash \varphi$ for each $\varphi \in \Phi$.
\end{itemize}
\end{definition}

\begin{lemma}
  The following statements hold.
  \begin{itemize}
  \item $V \not\Vdash \varphi$ if, and only if, $V \Vdash \neg\varphi$
  \item $V \Vdash \varphi$ if, and only if, $V \not\Vdash \neg\varphi$
  \end{itemize}
\end{lemma}

\begin{definition}\label{def:term}
  A formula $\varphi \in \fp$ is:
  \begin{itemize}
  \item \emph{satisfiable}, if $V \Vdash \varphi$, for some $V$;
  \item \emph{valid} (denoted $\models\varphi$), if $V \Vdash
    \varphi$, for all $V$;
  \item \emph{contradictory}, if no $V$ is such that $V \Vdash \varphi$.
  \end{itemize}
\end{definition}
One may write $\not\models\varphi$, if $\varphi$ is \emph{not}
valid. The notion of satisfiability also applies to sets of formulae:
a set $\Phi \subseteq \fp$ is \emph{satisfiable}, if there is a $V$
that satisfies every formula in $\Phi$; otherwise, the set is said to
be \emph{contradictory}.

\begin{lemma}
  A formula that is \emph{not}:
  \begin{itemize}
  \item valid, is either satisfiable or contradictory;
  \item contradictory, is either satisfiable or valid;
  \item satisfiable, is contradictory (as it cannot be valid).
  \end{itemize}
\end{lemma}

\begin{lemma}
  The \emph{negation} of a formula:
  \begin{itemize}
  \item valid, is contradictory;
  \item contradictory, is valid;
  \item satisfiable (not valid), is satisfiable.
  \end{itemize}
\end{lemma}

\paragraph{Logical equivalence.}
There are syntactically different formul\ae that evaluate to the
same value, \ie, are equivalent. To rigorously define the notion, we
introduce first the idea a formula resulting  from (or being a
semantic consequence of) a set of formul\ae.

\begin{definition}\label{def:models}
  Let $\Phi \subseteq \fp$ and $\varphi \in \fp$. One may say that a
  formula $\varphi$ is a \emph{semantic consequence} of a set of
  formul\ae\ $\Phi$, denoted by $\Phi \models \varphi$, if whenever $V
  \Vdash \Phi$ also $V \Vdash \varphi$.
\end{definition}

\begin{proposition}
    $\set{\varphi_1,\ldots,\varphi_n} \models \psi\textrm{ if, and
      only if,
    }\models (\varphi_1 \wedge \ldots \wedge \varphi_n) \to \psi$, for any $n \in \nat$.
\end{proposition}

\begin{lemma}
  The following statements hold.
  \begin{enumerate}
  \item $\set{\bot} \models \varphi$
  \item $\set{\varphi \wedge \psi} \models \varphi$ and $\set{\varphi
      \wedge \psi} \models \psi$
  \item $\set{\varphi} \models \varphi \vee \psi$ and $\set{\psi}
    \models \varphi \vee \psi$
  \end{enumerate}
\end{lemma}

\begin{definition}\label{def:equiv}
  The formulae $\varphi, \psi \in \fp$ are \emph{logically
    equivalent}, denoted by $\varphi \equiv \psi$, whenever $\set\varphi
  \models \psi$ if, and only if, $\set\psi \models \varphi$.
\end{definition}

\begin{theorem}
  The binary relation $\equiv$ on $\fp$ is a \emph{congruence relation}.
\end{theorem}

\paragraph{Conjunctive Normal Form.}

\begin{definition}\label{def:lit}
  A \emph{literal} is an atomic formula (said \emph{positive}) or the
  negation of an atomic formula (said \emph{negative}).
\end{definition}
Recall that $\top \eqabv \neg\bot$ (being thus a negative literal).

\begin{lemma}\label{lem:disjLit}
  A disjunction of literals $\bigvee_{i=1}^nL_i$, with $n \geq 1$, is valid
  if, and only if, there are $1 \leq i,j \leq n$ such that $L_i=\top$ or $L_i=\neg L_j$.
\end{lemma}

\begin{definition}\label{def:cnf}
  A formula $\varphi \in \fp$ is in \emph{Conjunctive Normal Form}, if
  it is a conjunction of disjunctions of literals.
\end{definition}
Consider a predicate $\opnsf{CNF}$ such that $\opnsf{CNF}(\varphi)$
holds if $\varphi$ is in conjunctive normal form.

\begin{lemma}
  A formula $\varphi \in \fp$ such that $\opnsf{CNF}(\varphi)$ is:
  \begin{itemize}
  \item valid, if all disjunctions are valid;
  \item contradictory, if some of the disjunctions are contradictory;
  \item satisfiable, otherwise.
  \end{itemize}
\end{lemma}
Any propositional formula is convertible in an equivalent formula in
conjunctive normal form.

\begin{theorem}
  For any formula $\varphi \in \fp$ there is a formula $\psi \in \fp$
  such that $\varphi \equiv \psi$ and moreover, $\opnsf{CNF}(\psi)$.
\end{theorem}


%% file: hornForm.tex
\section{Conversion to Horn Formula}\label{app:hf}

Any basic Horn clause may be presented as an implication (\cf~Lemma ~\ref{lem:HornClauses}).

\begin{lemma}
  Let $L$ be a positive literal.
  \begin{enumerate}
  \item $L \equiv \top \rightarrow L$
  \item $\bigvee_{i=1}^n \neg L_i \equiv (\bigwedge_{i=1}^n L_i)
    \rightarrow \bot$
  \item $\bigvee_{i=1}^n \neg L_i \vee L \equiv (\bigwedge_{i=1}^n L_i)
    \rightarrow L$
  \end{enumerate}
\end{lemma}

\begin{proof}
  We use below standard equivalences of Propositional Logic. 
  Recall that logical equivalence is a congruence relation.
  \begin{enumerate}
  \item $L \equiv \top \rightarrow L$
    \begin{align*}
      L & \equiv L \vee \bot \\
      & \equiv L \vee \neg\neg\bot \\
      & \equiv L \vee \neg\top \\
      & \equiv \neg\top \vee L \\
      & \equiv \top \to L
    \end{align*}
%
  \item $\bigvee_{i=1}^n \neg L_i \equiv (\bigwedge_{i=1}^n L_i) \to \bot$
    The proof is by natural induction, using the following law.
    $$(\varphi \to \gamma) \vee (\psi \to \gamma)
    \equiv (\varphi \wedge \psi) \to \gamma$$
    \begin{description}
    \item[Base case:] n=1.
      $$\bigvee_{i=1}^n \neg L_i  = \neg L_1 \equiv \neg L_1 \vee \bot \equiv L_1 \to \bot$$
    \item[Inductive step:]
      $$\bigvee_{i=1}^{n+1} \neg L_i  = \bigvee_{i=1}^{n} \neg L_i \vee
      \neg L_{n+1} \equiv ((\bigwedge_{i=1}^n L_i) \to \bot) \vee
      (L_{n+1} \to \bot) \equiv (\bigwedge_{i=1}^{n+1}L_i) \to \bot$$
    \end{description}

    The proof of the auxiliar law is easy.
    \begin{align*}
      (\varphi \wedge \psi) \to \gamma & \equiv \neg(\varphi \wedge \psi) \vee \gamma \\
      & \equiv (\neg\varphi \vee \neg\psi) \vee \gamma \\
      & \equiv (\neg\varphi \vee \neg\psi) \vee (\gamma \vee \gamma) \\
      & \equiv (\neg\varphi \vee \gamma) \vee (\neg\psi \vee \gamma) \\
      & \equiv (\varphi \to \gamma) \vee (\psi \to \gamma)
    \end{align*}
  \item $\bigvee_{i=1}^n \neg L_i \vee L \equiv (\bigwedge_{i=1}^n L_i) \to L$\\

    The proof is by natural induction.
    \begin{description}
    \item[Base case:] n=1.
      $$\bigvee_{i=1}^n \neg L_i \vee L  = \neg L_1 \vee L \equiv L_1 \to L$$
    \item[Inductive step:]
      \begin{align*}
        \bigvee_{i=1}^{n+1} \neg L_i \vee L & \equiv (\bigvee_{i=1}^n \neg L_i \vee \neg L_{n+1}) \vee (L \vee L) \\
        & \equiv (\bigvee_{i=1}^n \neg L_i \vee L) \vee (\neg L_{n+1} \vee L) \\
        & \equiv (\bigwedge_{i=1}^n L_i \to L) \vee (L_{n+1} \to L) \\
        & \equiv (\bigwedge_{i=1}^{n+1} L_i) \to L
      \end{align*}
    \end{description}
  \end{enumerate}
\end{proof}


%% file: lfp.tex
\section{Least Fixed-Point}\label{app:lfp}

We present here the proof of Lemma~\ref{lem:lfp}.

\begin{lemma}\label{app:lem:lfp}
  Let $\varphi$ be a Horn formula, \ie,
  $\varphi = \bigwedge_{i=1}^n(C_i \rightarrow L_i)$.
  The function $\mathcal A$ is:
  \begin{enumerate}
  \item\label{inc} increasing: $\mathcal C \subseteq \mathcal A(\varphi,\mathcal C)  \subseteq \mathcal C \cup \bigcup_{i=1}^n \set{L_i}$;
  \item\label{mon} and monotone: if $\mathcal C \subseteq \mathcal D$ then
    $\mathcal A(\varphi,\mathcal C) \subseteq \mathcal A(\varphi,\mathcal D)$.
  \end{enumerate}
\end{lemma}
\begin{proof}
  The proofs of both cases are so similar that we present them
  together. If $\mathcal A(\varphi,\mathcal C) = \mathcal C$, the
  results hold trivially. Otherwise, let
  $\mathcal C' = \mathcal A(\varphi,\mathcal C)$ and
  $\mathcal D' = \mathcal A(\varphi,\mathcal D)$.
  We proceed by natural induction on the number of clauses in $\varphi$.
  \begin{description}
  \item[Base case:] let $\varphi = C \to L$. Since $\opnsf{set}(C)
    \subseteq \mathcal C$ (as $\mathcal A(\varphi,\mathcal C) \not=
    \mathcal C$), then $\mathcal A(C \to L,\mathcal C) = \mathcal C
    \cup \set L$.  By hypothesis $\mathcal C \subseteq \mathcal D$,
    thus $\opnsf{set}(C) \subseteq \mathcal D$. Therefore, $\mathcal C
    \subseteq \mathcal C' = \mathcal C \cup \set L \subseteq \mathcal
    D \cup \set L = \mathcal D'$, and thus $\mathcal A$ is increasing
    and monotone.
    \item[Inductive step:] let
      $\varphi = C \to L \wedge \bigwedge_{i=1}^{n+1}(C_i \rightarrow L_i)$, 
      where $n \geq 0$. Assume, without loss of generality, that
      $\opnsf{set}(C) \subseteq \mathcal C$. Then,
      $$\mathcal C' = \mathcal A(\varphi,\mathcal C) =
        \mathcal A(\bigwedge_{i=1}^{n+1}(C_i \rightarrow L_i),\mathcal C \cup \set L)$$
        If $\mathcal A(\bigwedge_{i=1}^{n+1}(C_i \rightarrow L_i),\mathcal C \cup \set L)
             = \mathcal C \cup \set L$, the results hold trivially.
        Otherwise, by induction hypothesis, 
        \begin{enumerate}
        \item $\mathcal C \cup \set L \subseteq \mathcal
          A(\bigwedge_{i=1}^{n+1}(C_i \rightarrow L_i),\mathcal C \cup \set L) \subseteq
          \mathcal C \cup \set L \cup \bigcup_{i=1}^n \set{L_i}$;
        \item if $\mathcal C \cup \set L \subseteq \mathcal D \cup \set L$
          then $\mathcal A(\bigwedge_{i=1}^{n+1}(C_i
          \rightarrow L_i),\mathcal C \cup \set L) \subseteq \mathcal
          A(\varphi,\mathcal D \cup \set L)$.
        \end{enumerate}
        It is now simple to show the results. The function $\mathcal A$ is:
        \begin{description}
        \item increasing - $\mathcal C \subseteq \mathcal C \cup \set L
          \subseteq \mathcal C' \subseteq
          \mathcal C \cup \set L \cup \bigcup_{i=1}^n \set{L_i}$; and
        \item monotone - considering $\mathcal C \subseteq \mathcal D$,
          also $\mathcal C \cup \set L \subseteq \mathcal D \cup \set L$,
          and as  $\mathcal D \subseteq \mathcal D \cup \set L$, we
          conclude $\mathcal C' \subseteq \mathcal A(\varphi,\mathcal D)
          \subseteq \mathcal A(\varphi,\mathcal D \cup \set L)$.
        \end{description}
  \end{description}
\end{proof}


%% file: term.tex
\section{Termination and complexity}\label{app:term}

\paragraph{Auxiliary notions.}
Henceforth, let $\varphi = \bigwedge_{i=1}^n(C_i \rightarrow L_i)$,
where $n \geq 1$ be a Horn formula. Thus, each $\opnsf{set}(C_i)$ is a
set of positive literals.
Recall that a Horn formula may be regarded as a set of
clauses.\footnote{Any propositional formula in $\opnsf{CNF}$
  determines univocally a set of sets of literals.}  Whenever $\varphi
\in E_P$ is a Horn form such that $\varphi = \varphi_1 \wedge
\varphi_2$, we may write $\varphi_1 \subseteq \varphi$.  Then, for
$\psi \subseteq \varphi$ and $\mathcal C \subseteq \mathcal C'$, when
we write $\mathcal A(\varphi,\mathcal C) =^k \mathcal A(\psi,\mathcal
C')$, the equality '$=^k$' denotes that the term on the right is
obtained from the term on the left by executing $k$ steps of the
algorithm.

\paragraph{Main result.}
Theorem \ref{thm:term} is in fact a corollary of the following general result.

\begin{theorem}\label{appthm:term}
  For any Horn formula $\varphi = \bigwedge_{i=1}^n(C_i \to L_i)$ and
  any set $\mathcal C$ of literals such that
  $$\set\top \subseteq \mathcal C \subseteq \set\top \cup \bigcup_{i=1}^n L_i$$
  there is a unique set of literals $\mathcal C'$ such that
  $\mathcal A(\varphi,\mathcal C) = \mathcal C'$. Furthermore, the
  function $\mathcal A$ takes at most $n+1$ steps to yield
  $\mathcal C'$, where $n$ is the number of clauses of $\varphi$.
\end{theorem}
\begin{proof}
  We proceed by natural induction on the number of clauses in
  $\varphi$.\pagebreak
  \begin{description}
    \item[Base case:] since $\varphi$ is a single clause; then, either
      $C = \top$ or $C = \bigwedge_{i=1}^n L_i$.
      \begin{enumerate}
      \item Case $\varphi = \top \to L$; therefore, as $\top \in \mathcal
        C$ by hypothesis, it is the case
        that $$\mathcal A(\varphi,\mathcal C) =
        \mathcal A(\top,\mathcal C \cup \set L) = \mathcal C \cup
        \set L$$
      \item Case $\varphi = \bigwedge_{i=1}^n L_i \rightarrow L$; therefore, as $
        \mathcal A(\varphi,\mathcal C) = \mathcal A(\top,\mathcal C') = \mathcal C'$,
        where
        $$\mathcal C' = \left\{
          \begin{array}{ll}
            \mathcal C \cup \set L, & \textrm{if }\set{L_i\ |\
              \textrm{forall }1 \leq i \leq n} \subseteq \mathcal C\\
            \mathcal C, & \textrm{otherwise}
          \end{array}
        \right.
        $$
      \end{enumerate}
      In both cases the algorithm returns the result in two steps: one
      to analyse the clause and affect the resulting set; another to
      finish the execution, using the base case of the inductive
      definition. Notice that as $n=1$, the execution of $\mathcal A$
      takes exactly  $n+1=2$ steps.
    \item[Inductive step:] let $\varphi = \bigwedge_{i=1}^{n+1}(C_i
      \rightarrow L_i)$, where $n \geq 0$;
      notice that each $\set{C_i}$ is either $\set\top$ or a set of
      literals.  Considering $\psi = \bigwedge_{i=1}^n(C_i \rightarrow
      L_i)$, then $\varphi = \psi \wedge (C_{n+1} \to L_{n+1})$.
      Assume, without loss of generality, that one chooses $\psi$ such
      that
      $$\mathcal A(\varphi,\mathcal C) =^k
          \mathcal A(\psi' \wedge (C_{n+1} \to L_{n+1}),\mathcal C') =
          \mathcal C''$$
      where
      \begin{enumerate}
      \item $0 \leq k \leq n$;
      \item $\psi' \subseteq \psi$, \ie, it is a subset of clauses;
      \item $C'' = \left\{
          \begin{array}{ll}
            \mathcal C' \cup \set{L_i}, & \textrm{if }\set{C_i} \subseteq \mathcal C'\\
            \mathcal C', & \textrm{otherwise}
          \end{array}
        \right.$
      \end{enumerate}
      By induction hypothesis $\mathcal C'$ exists. Therefore,
      $\mathcal C''$ exists and is obtained from $\mathcal C'$ in two
      steps. Therefore, the execution of $\mathcal A$ takes $k+2$
      steps and $$k+2 \leq n+2 = (n+1)+1.$$
  \end{description}
\end{proof}


%% file: leastM.tex
\section{Unique least model}\label{app:leastM}
We present now the proof of Proposition~\ref{prop:leastM}, \ie, the
existance of an Herbrand model.

\medskip
Let $\smb\varphi$ denote the set of propositional symbols of the
formula $\varphi$, inductively defined on the productions generating
the Propositional Language (\cf Definition \ref{def:propL}).
Notice first the following simple fact. 

\begin{lemma}\label{applem:bottop}
  For any Horn formula $\varphi \in E_P$, let $\mathcal
  A(\varphi,\set\top) = \mathcal C$. Then
  $\mathcal C \subseteq \smb\varphi \cup \set{\bot,\top}$.
\end{lemma}

This result follows easily from Lemma~\ref{app:lem:lfp}.\ref{inc} and
Theorem~\ref{appthm:term}.

\begin{proposition}\label{appprop:leastM}
  Consider a satisfiable Horn formula $\varphi \in E_P$ such that
  $\mathcal A(\varphi,\set\top) = \mathcal C$ and
  $\bot \notin\mathcal C$. Then, $V \Vdash \varphi$ considering $V$
  such that $V(p)=1$ for each $p \in \mathcal C$ and $V(q)=0$ for each
  $q \in (P \setminus \mathcal C)$.
\end{proposition}
\begin{proof}
  Let
  $\mathcal C = \set\top \cup \set{p_i\ |\ 1 \leq i \leq r \textrm{,
      for some }r \geq 0} \subseteq \smb\varphi \cup \set{\bot,\top}$
  (by the previous lemma).  By Proposition~\ref{prop:HornForm},
  consider
  $$\varphi = \bigwedge_{i=1}^n(\top \rightarrow L_i)
        \wedge
        \bigwedge_{j=1}^m(C_j \rightarrow L_j)
  $$
  where each $C_j = \bigwedge_{k=1}^{l_j} L_{k,j}$, for some $l_i$,
  with all literals positive (so, no $L_{k,j}$ is $\top$).

  A valuation $V$ satisfies $\varphi$ if
  $V \Vdash \bigwedge_{i=1}^n(\top \rightarrow L_i)$ and
  $V \Vdash \bigwedge_{j=1}^m(C_j \rightarrow L_j)$. Obviously, for
  all propositional symbols $q \notin \smb\varphi$ we can have
  $V(q)=0$; so we consider below only symbols in
  $\smb\varphi$.
  \begin{description}
  \item[Case]$V \Vdash \bigwedge_{i=1}^n(\top \rightarrow L_i)$.
    Since all $L_i$ are positive literals and none can be $\bot$
    (otherwise $V$ could not satisfy the formula consider herein), all
    must be propositional symbols, say $p_i$. Therefore, by definition
    of $\mathcal A$ and by Lemma~\ref{app:lem:lfp} we have
    $\mathcal C' = \set{p_i\ |\ 1 \leq i \leq n} \subseteq \mathcal
    C$, and since by hypothesis $V \Vdash \top \rightarrow p_i$ for
    all $i \in \set{1,\ldots,n}$, it holds as envisaged that
    $V(p_i)=1$.
  \item[Case]$V \Vdash \bigwedge_{j=1}^m(C_j \rightarrow L_j)$ where
    each $C_j = \bigwedge_{k=1}^{l_j} L_{k,j}$, for some $l_j$, with
    all literals positive. Obviously, for all $j \in \set{1,\ldots,m}$,
    we have that $V \Vdash C_j \rightarrow L_j$ where each $L_j$
    might be either $\bot$ or a propositional symbol. We consider now
    both cases.

    If $L_j=\bot$ then $V \not\Vdash C_j$; so we need to consider two
    cases:
  \begin{enumerate}
  \item either some $L_{k,j}$ is $\bot$; or
  \item for some $k\in\set{1,\ldots,l_j}$ we have $V(L_{k,j})=0$; if
    $L_{k,j} \in \mathcal C$, then in the case that all other literals
    in $C_j$ that are not $\bot$ are also in $\mathcal C$, by
    definition of $\mathcal A$ we would get $\bot \in \mathcal C$,
    what contradicts the hypothesis; therefore,
    $L_{k,j} \notin \mathcal C$.
  \end{enumerate}
  Note that we can consider any valuation for the remaining literals
  which are propositional symbols.

  If $L_j$ is a propositional symbol (say $p$), then again we need to
  consider two cases.
  \begin{enumerate}
  \item either $V(p)=0$ and thus $V \not\Vdash C_j$, and we proceed as
    above; or
  \item $V(p)=1$ and thus $V \Vdash C_j$, \ie, for all
    $k\in\set{1,\ldots,l_j}$ we have $V(L_{k,j})=1$, with, by
    definition of the $\mathcal A$, all $L_{k,j}$ and $p$ in $\mathcal C$.
  \end{enumerate}
  \end{description}
\end{proof}

To prove an equivalent formulation of this result --- for a
satisfiable Horn formula $\varphi$ such that
$\mathcal A(\varphi,\set\top) = \mathcal C$, any $p \in \mathcal C$ if
and only if its valuation is 1 --- one might proceed axiomatically,
using the following rule.
$$(\bigwedge_{i=1}^n(\top \rightarrow p_i) \wedge
(\bigwedge_{i=1}^np_i \to p)) \to (\top \to p)$$ The satisfaction of
$\varphi$ implies the satisfaction of the formula above and thus, by
definition of the function $\mathcal A$ and of the satisfaction
relation (\cf Definition \ref{def:sat}), one easily concludes that
$p \in \mathcal C$ and its valuation is 1


%% file: corr.tex
\section{Correctness}\label{app:corr}

\subsection{Soundness}

\begin{lemma}\label{lem:topbot}
  Let $\varphi \in E_P$ be in Horn formula such that $\varphi = \psi \wedge
  (C \to \bot)$ and $\bot \in \mathcal A(\varphi,\set\top)$. Then,
  $\set\varphi \models (\top \to \bot)$, being thus contradictory.
\end{lemma}
\begin{proof}
  Using laws of Propositional Logic (in particular $(\varphi
  \rightarrow \psi) \wedge (\psi \rightarrow \gamma) \models \varphi
  \to \gamma$, one easily shows the result.
\end{proof}


\begin{theorem}\label{appthm:corr})
  For any Horn formula $\varphi \in E_P$:
  \begin{itemize}
  \item $\bot \notin \mathcal A(\varphi,\set\top)$, only if $\varphi$ it is satisfiable;
  \item $\bot \in \mathcal A(\varphi,\set\top)$, only if $\varphi$ it is contradictory.
  \end{itemize}
\end{theorem}
\begin{proof}
  The first statement is a consequence of Proposition \ref{prop:leastM}.
  We prove the second statement. Since by hypothesis, $\bot \in
  \mathcal A(\varphi,\set\top)$, the contra-positive of Lemma
  \ref{cor:opt1} ensures that either $\varphi = \top \to \bot$ or
  there is a Horn formula $\gamma$ such that $\varphi = \gamma \wedge
  (C_i \to \bot)$, for some $1 \leq i \leq n$. The case $\varphi =
  \top \to \bot$ yields immediatly the result, as $\varphi \equiv
  \bot$. Let us then consider the other case.\\[.5em]

  \noindent Let $\mathcal A(\varphi,\set\top) =^k \mathcal
  A(\psi,\mathcal C)$ with:
  \begin{enumerate}
  \item $0 \leq k < n$;
  \item $\top \in \mathcal C$ (by Lemma \ref{lem:incr}) and $\bot
    \notin \mathcal C$;
  \item $\varphi = \varphi' \wedge \psi$ and either $\psi = \top \to
    \bot$ or $\psi = \psi' \wedge (C_i \to \bot)$, for some Horn form
    $\psi'$.
  \end{enumerate}
  Assume, without loss of generality, that $\set{C_i} \subseteq
  \mathcal C$; then, by Lemma \ref{lem:incr},
  $$\mathcal A(\psi,\mathcal C) = \mathcal A(\psi',\mathcal C \cup
  \set\bot) \subseteq \mathcal C \cup \set\bot$$ Since $\varphi =
  \varphi' \wedge \psi' \wedge (C_i \to \bot)$, by Lemma
  \ref{lem:topbot} we conclude that $\varphi$ it is contradictory.
\end{proof}

\subsection{Completeness}

\begin{theorem}\label{appthm:compl}
  For any Horn formula $\varphi \in E_P$:
  \begin{itemize}
  \item $\bot \notin \mathcal A(\varphi,\set\top)$, if $\varphi$ it is satisfiable;
  \item $\bot \in \mathcal A(\varphi,\set\top)$, if $\varphi$ it is contradictory.
  \end{itemize} 
\end{theorem}
\begin{proof}
  The first statement is the contra-positive of the second statement of
  the previous theorem. The second is the contra-positive of the
  first statement of the previous theorem.
\end{proof}
